\documentclass[11pt]{amsart}
\usepackage{amsmath}
\usepackage[dvips]{epsfig}
\theoremstyle{plain}
\newtheorem{theorem}{Theorem}[section]
\newtheorem{lemma}[theorem]{Lemma}

\theoremstyle{definition}

\numberwithin{equation}{section}

\begin{document}

\title[A Penrose-Like Inequality with Charge]
{A Penrose-Like Inequality with Charge}

\author[Khuri]{Marcus A. Khuri}
\address{Department of Mathematics\\
Stony Brook University\\ Stony Brook, NY 11794}
\email{khuri@math.sunysb.edu}
\thanks{The author was partially supported by NSF Grants DMS-1007156 and DMS-1308753.}
\begin{abstract}
We establish a Penrose-like inequality for general (not
necessarily time-symmetric) initial data sets of the Einstein-Maxwell
equations, which satisfy the dominant energy condition.  More
precisely, it is shown that the ADM energy is bounded below by an
expression which is proportional to the sum of the square root of the area of
the outermost future (or past) apparent horizon and the square of the total charge.
The proportionality constants depend on the solution to a linear elliptic equation
which incorporates the charge. In addition, a corrected version of the Penrose-like inequality in \cite{Khuri}
is presented.
\end{abstract}
\maketitle

\section{Introduction}\label{sec1}
\setcounter{equation}{0}

Consider an initial data set $(M, g, k, E)$ for the Einstein-Maxwell equations with vanishing magnetic field. Here $M$ is a Riemannian $3$-manifold with metric $g$, $k$ is a symmetric 2-tensor representing the second fundamental
form of the embedding into spacetime, and $E$ denotes the electric field. It is assumed that the manifold has a boundary $\partial M$ consisting of an outermost apparent horizon. That is, if $H$ denotes mean curvature with respect to the normal pointing towards spatial infinity, then each boundary component $S\subset\partial M$ satisfies $\theta_{+}(S):=H_{S}+Tr_{S}k=0$ (future horizon) or $\theta_{-}(S):=H_{S}-Tr_{S}k=0$ (past horizon), and
there are no other apparent horizons present. Moreover the data are taken to be asymptotically flat with one end, in that
outside a compact set the manifold is diffeomorphic to the complement of a ball in $\mathbb{R}^{3}$, and in the coordinates given
by this asymptotic diffeomorphism the following fall-off conditions hold
\begin{equation*}\label{1}
|\partial^{m}(g_{ij}-\delta_{ij})|=O(|x|^{-m-1}),\text{ }\text{
}\text{ }|\partial^{m}k_{ij}|=O(|x|^{-m-2}),\text{ }\text{
}\text{ }|\partial^{m}E^{i}|=O(|x|^{-m-2}),\text{ }\text{ }\text{ }m=0,1,2,\text{ }\text{ }\text{as}\text{ }\text{
}|x|\rightarrow\infty.
\end{equation*}
With a vanishing magnetic field, the matter and current densities for the non-electromagnetic matter fields are given by
\begin{align}\label{2}
\begin{split}
 2 \mu  & = R + (Tr k)^2 - |k|_{g}^2 - 2|E|_{g}^2, \\
J & = div (k - (Tr k)g),
\end{split}
\end{align}
where $R$ denotes the scalar curvature of $g$. The following inequality will be referred to as the dominant energy condition
\begin{equation}\label{3}
\mu \geq |J|_{g}.
\end{equation}
Note that this dominant energy condition differs from the standard one, in that the energy density for the electric field is removed.
Under these hypotheses and
based on heuristic arguments of Penrose \cite{Penrose} which rely heavily on the cosmic censorship conjecture, the following inequality relating the ADM energy and the minimal area $\mathcal{A}$ required to enclose the
boundary $\partial M$, has been conjectured to hold
\begin{equation}\label{4}
E_{ADM}\geq\sqrt{\frac{\mathcal{A}}{16\pi}}+\sqrt{\frac{\pi}{\mathcal{A}}}Q^{2},
\end{equation}
where $Q=\lim_{r\rightarrow\infty}\frac{1}{4\pi}\int_{S_{r}}E^{i}\nu_{i}$ is the total electric charge, with $S_{r}$ coordinate spheres in
the asymptotic end having unit outer normal $\nu$. Inequality \eqref{4}
has been proven by Jang \cite{Jang1} for time-symmetric
initial data with a connected horizon, under the assumption that a smooth solution to the Inverse Mean Curvature Flow (IMCF)
exists. Moreover in light of Huisken and Ilmanen's work \cite{HuiskenIlmanen}, the hypothesis of a smooth IMCF can be discarded.
However without the assumption of a connected horizon, counterexamples \cite{WeinsteinYamada} are known to exist (these examples do not provide a contradiction to the
cosmic censorship conjecture), although \eqref{4} remains true \cite{Khuri3} if an auxiliary inequality holds between the area and charge.
In the non-time-symmetric case this inequality has been proven under the additional hypothesis of spherically symmetric initial data \cite{Hayward}.
In the general case, with a connected horizon, the validity of \eqref{4} has been reduced to solving a coupled system of equations involving
the generalized Jang equation and the IMCF \cite{DisconziKhuri}. In the case of equality, it is expected that the initial data arise from the Reissner-Nordstr\"{o}m
spacetime; this has been confirmed in the time-symmetric case \cite{DisconziKhuri}.

In this paper we establish a Penrose-like inequality including charge, without any assumption on $k$ or on the connectedness of the boundary. The primary difficulty in the
non-time-symmetric (and non-maximal) case is the lack of the following positive lower bound for the scalar curvature
\begin{equation}\label{5}
R\geq2|E|_{g}^{2}.
\end{equation}
In order to circumvent this issue, we seek a deformation of the initial data to a new set
$(\overline{M},\overline{g},\overline{E})$, where $\overline{M}$ is diffeomorphic to $M$, and the metric $\overline{g}$ and vector field
$\overline{E}$ are related to $g$ and $E$ in a precise way described below. The purpose of the deformation is to obtain new initial
data which satisfy \eqref{5} in a weak sense, while preserving the relevant geometric and physical quantities,
such as the charge density, total charge, ADM energy, and boundary area. The desired deformation is a generalization of a procedure introduced by Jang \cite{Jang} and studied
extensively by Schoen and Yau \cite{SchoenYau}. More precisely, consider the product 4-manifold $(M \times \mathbb{R}, g + dt^2)$, and let $\overline{M}=\{t=f(x)\}$ be
the graph of a function $f$ inside this setting. Then the induced metric on $\overline{M}$ is given by $\overline{g}=g+df^{2}$. In order to obtain the most desirable positivity property
for the scalar curvature of the graph, the function $f$ should satisfy
\begin{equation}\label{6}
 \left( g^{ij} - \frac{ f^i  f^j}{1 + |\nabla f|_g^2 }\right)
\left( \frac{ \nabla_{ij}f}{ \sqrt{1 + |\nabla f|_g^2 }}
-k_{ij} \right) = 0,
\end{equation}
where $\nabla$ denotes covariant differentiation with respect to the metric $g$, $f_{i}=\partial_{i}f$,
and $f^{i}=g^{ij}f_{j}$. Equation \eqref{6} is referred to as the Jang equation, and when it is
satisfied $\overline{M}$ will be called the Jang surface.
The scalar curvature of the Jang surface \cite{SchoenYau} is given by
\begin{equation}\label{7}
\overline{R}=2(\mu-J(w))+2|E|_{g}^{2}+
|h-k|_{\overline{g}}^{2}+2|q|_{\overline{g}}^{2}
-2\overline{div}(q),
\end{equation}
here $\overline{div}$ is the divergence operator with respect to $\overline{g}$,
$h$ is the second fundamental form of the graph $t=f(x)$ in the Lorentzian 4-manifold $(\overline{M} \times \mathbb{R}, \overline{g}-dt^2)$,
and $w$ and $q$ are 1-forms given by
\begin{equation}\label{8}
h_{ij}=\frac{ \nabla_{ij}f}{ \sqrt{1 + |\nabla f|_g^2 }},\text{ }\text{ }\text{ }\text{ }
w_{i}=\frac{f_{i}}{\sqrt{1+|\nabla f|_{g}^{2}}},\text{
}\text{ }\text{ }\text{ }
q_{i}=\frac{f^{j}}{\sqrt{1+|\nabla f|_{g}^{2}}}(h_{ij}-k_{ij}).
\end{equation}

The existence and regularity theory for equation \eqref{6} is well-understood. In particular, it is shown in \cite{HanKhuri} and \cite{Metzger} that there exists a smooth solution on $M$
which blows-up in the form of a cylinder over the outermost apparent horizon, with $f(x)\rightarrow\infty$ ($-\infty$) at each component of $\partial M$ depending on whether it is a future (or
past) apparent horizon. Let $\tau(x)=dist(x,\partial M)$, and denote the level sets of $\tau$ by $S_{\tau}$. If $|\theta_{\pm}(S_{\tau})|\sim\tau^{l}$ near a future (past) apparent horizon component of the boundary,
then according to \cite{HanKhuri} the blow-up solution satisfies the following asymptotics near that boundary component
\begin{equation}\label{8.1}
\alpha^{-1}\tau^{-\frac{l-1}{2}}+\beta^{-1} \leq \pm f \leq \alpha\tau^{-\frac{l-1}{2}}+\beta,
\end{equation}
for some positive constants $\alpha$ and $\beta$. Moreover the solution decays sufficiently fast at spatial infinity so that the ADM energies agree
$E_{ADM}(\overline{g})=E_{ADM}(g)$.

When the dominant energy condition is satisfied,
all terms appearing on the right-hand side of \eqref{7} are
nonnegative, except possibly the last term. Thus the scalar curvature is nonnegative modulo a divergence, so it may be described as weakly nonnegative.
For the topic of interest here, a stronger condition than simple nonnegativity is required, more precisely
we seek an inequality (holding in the weak sense) of the following form
\begin{equation}\label{9}
\overline{R}\geq 2|\overline{E}|_{\overline{g}}^{2},
\end{equation}
where $\overline{E}$ is an auxiliary electric field defined on the Jang surface. This auxiliary electric field is required to satisfy three properties,
namely
\begin{equation}\label{10}
|E|_{g}\geq|\overline{E}|_{\overline{g}},\text{ }\text{ }\text{ }\text{ }\text{ }\text{ }\overline{div}\,\overline{E}=0,\text{ }\text{ }\text{ }\text{ }\text{ }\text{ }\overline{Q}=Q,
\end{equation}
where $\overline{Q}$ is the total charge defined with respect to $\overline{E}$. In particular, if the first inequality of \eqref{10} is satisfied, then the dominant
energy condition \eqref{3} and the scalar curvature formula \eqref{7} imply that \eqref{9} holds weakly. It turns out that there is a very natural choice for
this auxiliary electric field, namely $\overline{E}$ is the induced electric field on the Jang surface $\overline{M}$ arising from the field strength $F$ of the electromagnetic
field on $(M \times \mathbb{R},  g + dt^2)$. More precisely $\overline{E}_i = F(N,X_i)$,
where $N$ and $X_i$ are respectively the unit normal and canonical tangent vectors to $\overline{M}$
\begin{equation}\label{11}
N=\frac{\partial_{t}-f^{i}\partial_{i}}{\sqrt{1+|\nabla f|_{g}^{2}}},\text{ }\text{ }\text{ }\text{ }\text{ }\text{ }X_{i}=\partial_{i}+f_{i}\partial_{t},
\end{equation}
and $F = \frac{1}{2} F_{ab}dx^a\wedge dx^b$ is given by
$F_{0i} = E_i$ and $F_{ij} = 0$ for $i=1,2,3$, with $x^i$, $i=1,2,3$ coordinates on $M$ and $x^{0}=t$. In matrix form
\begin{equation}\label{12}
F = \left(
\begin{array}{cccc}
  0 & E_1 & E_2 & E_3 \\
-  E_1 & 0 & 0 & 0 \\
-  E_2 & 0 & 0 & 0 \\
-  E_3 & 0 & 0 & 0
 \end{array}
\right).
\end{equation}
In \cite{DisconziKhuri} it is shown that
\begin{equation}\label{13}
\overline{E}_i = \frac{E_i +  f_i f^j E_j}{\sqrt{1 +  |\nabla f|^2_g}},
\end{equation}
and that all the desired properties of \eqref{10} hold. This auxiliary electric field was also used in \cite{Khuri2}.

The fact that inequality \eqref{9} holds in a weak sense, allows us to find (a proof is given in the next section) a unique positive solution to the prescribed scalar curvature equation
\begin{equation}\label{14}
\overline{\Delta}u-\frac{1}{8}\overline{R}u+\frac{1}{4}|\widehat{E}|_{\widehat{g}}^{2}u^{5}=0\text{ }\text{ }\text{ on }\text{ }\text{ }\overline{M}
\end{equation}
with the following boundary conditions. Namely, $u$ vanishes asymptotically along the cylindrical ends of $\overline{M}$ or alternatively $u(x)\rightarrow 0$ as $x\rightarrow\partial M$, and
\begin{equation}\label{15}
u(x)=1+\frac{A}{|x|}+O\left(\frac{1}{|x|^{2}}\right)\text{ }\text{ }\text{ as }\text{ }\text{ }|x|\rightarrow\infty
\end{equation}
for some constant $A$. Here $\overline{\Delta}$ is the Laplacian with respect to $\overline{g}$, $\widehat{g}=u^{4}\overline{g}$, and $\widehat{E}^{i}=u^{-4}\overline{E}^{i}$.
Equation \eqref{14} expresses the fact that the scalar curvature of $\widehat{g}$ is given by
\begin{equation}\label{16}
\widehat{R}=2|\widehat{E}|_{\widehat{g}}^{2}.
\end{equation}
It follows that the conformally deformed initial data $(\overline{M},\widehat{g},\widehat{E})$ satisfies the desired version \eqref{5} of the dominant energy condition. We point out that
the process of conformally changing the Jang initial data in order to obtain favorable properties for the scalar curvature was first used by Schoen and Yau \cite{SchoenYau} in their proof
of the positive mass theorem. In fact when $E=0$, the solution $u$ of \eqref{14} coincides with the conformal factor used in \cite{SchoenYau}.

We now state the main theorem. Recall that the Hawking mass of a surface $S\subset M$, with area $|S|$, is given by
\begin{equation}\label{16.1}
M_{H}(S)=\sqrt{\frac{|S|}{16\pi}}\left(1-\frac{1}{16\pi}\int_{S}H^{2}\right),
\end{equation}
and that $S$ is said to be area outerminimizing if every surface which encloses it has area greater than or equal to $|S|$. If a connected surface $S$ encloses
the boundary $\partial M$, the region between $S$ and spatial infinity will be denoted by $M_{S}$.

\begin{theorem}\label{thm1}
Let $(M,g,k,E)$ be a smooth asymptotically
flat initial data set for the Einstein-Maxwell equations with total charge $Q$, $div E=0$, and satisfying the
dominant energy condition $\mu\geq|J|_{g}$.  If the boundary consists
of an outermost apparent horizon with components
$\partial_{i}M$, $i=1,\ldots,n$, then
\begin{equation}\label{17}
E_{ADM}(g)\geq\frac{\sigma_{1}}{2(1+\sigma_{1})}\sum_{i=1}^{n}\sqrt{\frac{|\partial_{i}
M|_{g}}{\pi}} +\sigma_{2}\sqrt{\frac{\pi}{|\partial M|_{g}}}Q^{2}
\end{equation}
\textit{with}
\begin{equation}\label{18}
\sigma_{1}=\left(\sum_{i=1}^{n}\sqrt{4\pi|\partial_{i}M|_{g}}\right)^{-1}\parallel\overline{\nabla}
u\parallel^{2}_{L^{2}(\overline{M})},\text{ }\text{ }\text{ }\text{ }
\sigma_{2}=\sup_{S}\sqrt{\frac{|\partial M|_{g}}{|S|_{\widehat{g}}}}\min_{M_{S}}u^{4},
\end{equation}
where the supremum is taken over all connected surfaces $S$ which enclose $\partial M$, are area outerminimizing,
and have nonnegative Hawking mass all with respect to $\widehat{g}$.
\end{theorem}

Note that the constants $\sigma_{1}$, $\sigma_{2}$ are scale invariant making them independent of $|\partial M|_{g}$, and it is clear that $\sigma_{1}$ is strictly positive. It will be proven below that $\sigma_{2}$ is also strictly positive. Thus inequality \eqref{17} has a similar structure to that of \eqref{4}, and it applies in a more general setting without restriction on the number of boundary components. This theorem also
applies without the assumption of time-symmetric or maximal data, whereas \eqref{4} has so far only been confirmed with these added hypotheses. It turns out that the case of equality in \eqref{17} cannot occur, which indicates that this inequality is not optimal.
Lastly, similar Penrose-like inequalities have previously been discussed in \cite{Herzlich} and \cite{Khuri}, without a contribution from the total charge. Issues with \cite{Herzlich} have been raised in \cite{BartnikChrusciel} and \cite{Malec} (and partially addressed in \cite{BartnikChrusciel}), while issues with \cite{Khuri} have been pointed out in \cite{Khuri1} and are resolved in the appendix of the present
paper.

We remark that it is the special geometry of the Jang surface, namely that it blows-up as a cylinder over the horizon, which is responsible for a definite contribution of area from each
boundary component to the right-hand side of \eqref{17}. This will be examined in Section \ref{sec4}. There it will also be shown that the constant $\sigma_{1}$ may be written as an infimum over
all functions satisfying appropriate asymptotics.

\section{The Conformal Factor}\label{sec2}

In the work of Schoen and Yau \cite{SchoenYau} existence of a unique solution to the following boundary value problem was
established:
\begin{equation}\label{19}
\overline{\Delta}z-\frac{1}{8}\overline{R}z=0\text{
}\text{ }\text{ on }\text{ }\text{ }\overline{M},
\end{equation}
with $z(x)\rightarrow 0$ as $x\rightarrow\partial M$ and $z(x)\rightarrow 1$ as $|x|\rightarrow\infty$.
The inequality (4.6) in \cite{SchoenYau} shows
that the first eigenvalue, $\eta_{i}$, of the operator $\Delta-\frac{1}{8}K$ on $\partial_{i}M$, is strictly positive (here $K$ denotes Gaussian curvature). As observed by
Schoen and Yau, $z\sim e^{\mp\sqrt{\eta_{i}}t}\zeta_{i}(y)$, that is the conformal factor $z$ is asymptotic to $e^{\mp\sqrt{\eta_{i}}t}\zeta_{i}(y)$ depending on whether
the Jang surface blows up or down, where $\zeta_{i}$ is the corresponding first eigenfunction.

The same methods of \cite{SchoenYau} may also be used to establish the existence of a unique solution to the following boundary value problem:
\begin{equation}\label{19}
\overline{\Delta}u-\frac{1}{8}\overline{R}u+\frac{1}{4}|\overline{E}|^{2}_{\overline{g}}u=0\text{
}\text{ }\text{ on }\text{ }\text{ }\overline{M},
\end{equation}
with
\begin{equation}\label{19.1}
u(x)\rightarrow 0\text{ }\text{ }\text{ as }\text{ }\text{ }x\rightarrow\partial M,\text{ }\text{ }
\text{ and }\text{ }\text{ }u(x)\rightarrow 1\text{ }\text{ }\text{ as }\text{ }\text{ }|x|\rightarrow\infty.
\end{equation}
Note that equation \eqref{19} is equivalent to equation \eqref{14}.
In light of \eqref{7} and the dominant energy condition \eqref{3}, a slightly modified version of (4.6) in \cite{SchoenYau} shows
that the first eigenvalue, $\lambda_{i}$, of the operator $\Delta-\frac{1}{8}K+\frac{1}{4}(E\cdot n)^{2}$ on $\partial_{i}M$, is strictly positive (here $n$ denotes the
unit normal to $\partial_{i}M$) if the initial data are slightly perturbed so that $\mu>|J|_{g}$ at $\partial M$. Moreover as in \cite{SchoenYau}, $u\sim e^{\mp\sqrt{\lambda_{i}}t}\phi_{i}(y)$,
where $\phi_{i}$ is the corresponding first eigenfunction.

For the purposes of the proof of Theorem \ref{thm1}, it will be convenient to consider auxiliary boundary value problems, for which the solutions $u_{T}$ will converge to $u$; this will also
yield an alternate proof of existence for $u$.
In order to describe the auxiliary problems, for each $T>0$ let $\overline{M}_{T}$ denote the
portion of the Jang surface $\overline{M}$ which lies between the hyperplanes $t=\pm T$.
Let $\chi_{T}(y)$ denote the one parameter family of functions defined on a given boundary component $\partial_{i}M$ as the restriction
of $|q|_{\overline{g}}$ to $\partial_{i}\overline{M}_{T}$. According to the parametric estimates for the Jang equation \cite{SchoenYau}, the sequence of
functions $\chi_{T}$ is uniformly bounded and equicontinuous. Therefore after passing to a subsequence (still denoted by $\chi_{T}$ for convenience)
we have that $\chi_{T}\rightarrow\chi$ as $T\rightarrow\infty$, for some continuous function $\chi$.
There are two cases to consider, namely, case 1 when $\chi$ vanishes identically, and case 2 when $\chi$ does not vanish identically.

We will slightly perturb $u$ in order to prescribe appropriate boundary conditions
on certain cylindrical ends. For large $T$ and $T_{0}$ ($T>T_{0}$), let
$(\overline{M}_{T}-\overline{M}_{T_{0}})_{i}$ denote the component of $\overline{M}_{T}-\overline{M}_{T_{0}}$ associated with the boundary component $\partial_{i}M$.
Let $i=1,\ldots,m$ index the boundary components which fall under case 1, and
let $i=m+1,\ldots,n$ index the boundary components which fall under case 2. Set $\widehat{M}_{T}=\overline{M}-\bigcup_{i=1}^{m}(\overline{M}-\overline{M}_{T})_{i}$,
that is, $\widehat{M}_{T}$ is the Jang surface after the cylindrical ends corresponding to case 1 have been removed. Consider the boundary value problem
\begin{equation}\label{20}
\overline{\Delta}u_{T}-\frac{1}{8}\overline{R}u_{T}+\frac{1}{4}|\overline{E}|^{2}_{\overline{g}}u_{T}=0\text{
}\text{ }\text{ on }\text{ }\text{ }\widehat{M}_{T},
\end{equation}
\begin{equation}\label{21}
\partial_{\overline{N}}u_{T}+\frac{1}{4}\overline{H}u_{T}=\frac{1}{4}\sqrt{\frac{16\pi}{|\partial_{i}\overline{M}_{T}|_{\widehat{g}_{T}}}}u_{T}^{3}
\text{ }\text{ }\text{ on }\text{ }\text{ }\partial_{i}\overline{M}_{T},\text{ }\text{ }i=1,\ldots,m,
\end{equation}
\begin{equation}\label{22}
u_{T}(x)\rightarrow 0\text{ }\text{ }\text{ as }\text{ }\text{ }x\rightarrow\partial_{i} M,\text{ }\text{ }i=m+1,\ldots,n,
\text{ }\text{ }
u_{T}(x)\rightarrow 1\text{ }\text{ }\text{ as }\text{ }\text{ }|x|\rightarrow\infty,
\end{equation}
where the unit normal $\overline{N}$ (with respect to $\overline{g}$) points towards spatial infinity and $\widehat{g}_{T}=u_{T}^{4}\overline{g}$. Note that the boundary condition \eqref{21} expresses the
fact that the mean curvature of $\partial_{i}\overline{M}_{T}$, $i=1,\ldots,m$, with respect to $\widehat{g}_{T}$, is given by $\widehat{H}=\sqrt{\frac{16\pi}{|\partial_{i}\overline{M}_{T}|_{\widehat{g}_{T}}}}$.
The solutions $u_{T}$ to this problem approximate the solution $u$ of \eqref{19} for large $T$, as is shown in Theorem \ref{thm3} below. Furthermore,
as in \cite{SchoenYau} a separation of variables argument can be used to show that the solution $u_{T}$ possesses the same asymptotics as $u$ along the ends corresponding to $\partial_{i}M$, $i=m+1,\ldots,n$, namely
\begin{equation}\label{22.1}
u_{T}\sim e^{\mp\sqrt{\lambda_{i}}t}\phi_{i}(y).
\end{equation}

\begin{theorem}\label{thm2}
If $T$ is sufficiently large, then there exists a smooth positive solution to boundary value problem \eqref{20}, \eqref{21}, \eqref{22}.
\end{theorem}

\begin{proof}
Consider the functional
\begin{align}\label{23}
\begin{split}
P(v)=&\frac{1}{2}\int_{\widehat{M}_{T}}\left(|\overline{\nabla}v|^{2}+\frac{1}{8}\overline{R}(1+v)^{2}-\frac{1}{4}|\overline{E}|^{2}_{\overline{g}}(1+v)^{2}\right)\\
&-\sum_{i=1}^{m}\frac{1}{8}\int_{\partial_{i}\overline{M}_{T}}\overline{H}(1+v)^{2}
 +\sum_{i=1}^{m}\frac{\sqrt{\pi}}{2}\left(\int_{\partial_{i}\overline{M}_{T}}(1+v)^{4}\right)^{1/2}
\end{split}
\end{align}
on the space of functions
\begin{equation}\label{24}
\mathcal{W}=\{v\in W_{loc}^{1,2}(\widehat{M}_{T})\mid |x|^{j-1}\overline{\nabla}^{j}v\in L^{2}(\widehat{M}_{T}),\text{ }j=0,1,\text{ }
1+v\in W^{1,2}_{0}(\widehat{M}_{T})\},
\end{equation}
where $W^{1,2}_{0}(\widehat{M}_{T})$ is the closure, in the $W^{1,2}$-norm, of the space of smooth functions which have compact support when restricted to each cylindrical
end indexed by $i=m+1,\ldots,n$. Here $W^{1,2}\subset L^{2}$ is the space of functions with square integrable first derivatives.

In order to establish the existence (as well as
the asymptotic behavior) of a solution $v_{T}\in
\mathcal{W}\cap C^{\infty}(\widehat{M}_{T})$, it
is enough, by the arguments of \cite{Herzlich}, to show that for $T$
sufficiently large the functional $P$ is nonnegative. To see that this is the case,
use formula \eqref{7} and \eqref{10}, and integrate the divergence term by parts to
find that for any $v\in \mathcal{W}$,
\begin{align}\label{25}
\begin{split}
P(v)\geq&\int_{\widehat{M}_{T}}\left(\frac{3}{8}|\overline{\nabla}v|^{2}+\frac{1}{8}(\mu-|J|_{g})(1+v)^{2}\right)+
\sum_{i=1}^{m}\frac{\sqrt{\pi}}{2}\left(\int_{\partial_{i}\overline{M}_{T}}(1+v)^{4}\right)^{1/2}\\
&
-\sum_{i=1}^{m}\frac{1}{8}\int_{\partial_{i}\overline{M}_{T}}(\overline{H}-q(\overline{N}))(1+v)^{2}.
\end{split}
\end{align}
According to \cite{BrayKhuri} (also \cite{BrayKhuri1}) $\overline{H}\rightarrow 0$ as $T\rightarrow\infty$, and since the boundary components $\partial_{i}\overline{M}_{T}$ belong
to case 1 we have that $q(\overline{N})\rightarrow 0$ as $T\rightarrow\infty$. Moreover, the area of $\partial_{i}\overline{M}_{T}$ approximates the area
of $\partial_{i}M$. It then follows from Jensen's Inequality
\begin{equation}\label{26}
\left(\int_{\partial_{i}\overline{M}_{T}}(1+v)^{2}\right)^{2}\leq
|\partial_{i}\overline{M}_{T}|_{\overline{g}}\int_{\partial_{i}\overline{M}_{T}}(1+v)^{4},
\end{equation}
that for $T$ sufficiently large $P$ is nonnegative.

It remains to show that $u_{T}=1+v_{T}$ is strictly positive.
So suppose that $u_{T}$ is not positive and let $D_{-}$ be the
domain on which $u_{T}<0$.  Since $u_{T}\rightarrow 1$ as
$|x|\rightarrow\infty$, the closure of $D_{-}\cap\overline{M}_{T}$ must be compact.
Now multiply equation \eqref{20} through by $u_{T}$ and integrate by
parts to obtain
\begin{equation}\label{27}
\int_{D_{-}}|\overline{\nabla}u_{T}|^{2}\leq 0.
\end{equation}
Note that if $D_{-}\cap\partial_{i}\overline{M}_{T}\neq\emptyset$, $i=1,\ldots,m$
then the same arguments used above to show that $P$ is
nonnegative, must be employed.  It follows that $u_{T}\geq 0$.  To
show that $u_{T}>0$, one need only apply Hopf's maximum principle
(the boundary condition of \eqref{21} must be used to obtain this
conclusion at $\partial_{i}\overline{M}_{T}$, $i=1,\ldots,m$).
\end{proof}

Multiply equation \eqref{20} by $u_{T}=1+v_{T}$ and integrate by parts to obtain
\begin{align}\label{28}
\begin{split}
\mathcal{P}(v_{T})&:=\lim_{r\rightarrow\infty}\frac{1}{2}\int_{|x|=r}u_{T}\partial_{\overline{N}}u_{T}\\
&\geq\int_{\overline{M}_{T}}\frac{1}{4}|\overline{\nabla}v_{T}|^{2}
+\left(\frac{1}{8}(\mu-|J|)+\frac{1}{16}|q|_{\overline{g}}^{2}\right)(1+v_{T})^{2}\\
& +\int_{\partial\overline{M}_{T}}\frac{1}{8}q(\overline{N})(1+v_{T})^{2}+\frac{1}{2}u_{T}\partial_{\overline{N}}u_{T}.
\end{split}
\end{align}
A standard formula yields
\begin{equation}\label{29}
\partial_{\overline{N}}u_{T}=\frac{1}{4}\widehat{H}u_{T}^{3}-\frac{1}{4}\overline{H}u_{T},
\end{equation}
where $\widehat{H}$ and $\overline{H}$ are the mean curvatures with respect to $\widehat{g}_{T}$ and $\overline{g}$, respectively. It
follows that
\begin{align}\label{30}
\begin{split}
\mathcal{P}(v_{T})
&\geq\int_{\overline{M}_{T}}\frac{1}{4}|\overline{\nabla}v_{T}|^{2}
+\left(\frac{1}{8}(\mu-|J|)+\frac{1}{16}|q|_{\overline{g}}^{2}\right)(1+v_{T})^{2}\\
& +\int_{\partial\overline{M}_{T}}\frac{1}{8}(q(\overline{N})-\overline{H})(1+v_{T})^{2}+\frac{1}{8}\widehat{H}(1+v_{T})^{4}.
\end{split}
\end{align}
The quantity $\mathcal{P}(v_{T})$ appears in the formula for the ADM energy of the metric $\widehat{g}_{T}$, more precisely
$E_{ADM}(\widehat{g}_{T})=E_{ADM}(\overline{g})-\pi^{-1}\mathcal{P}(v_{T})$.
Therefore it is important to estimate $\mathcal{P}(v_{T})$ from below.

\begin{lemma}\label{lemma1}
Let $u_{T}=1+v_{T}$ be the function produced in Theorem \ref{thm2}, then
\begin{equation}\label{31}
\mathcal{P}(v_{T})\geq \int_{\overline{M}_{T}}\frac{1}{4}|\overline{\nabla}v_{T}|^{2}
+\left(\frac{1-\vartheta_{T}}{2}\right)
\sum_{i=1}^{n}\sqrt{\frac{\pi}{|\partial_{i}\overline{M}_{T}|_{\overline{g}}}}
\int_{\partial_{i}\overline{M}_{T}}(1+v_{T})^{2}
\end{equation}
where $\vartheta_{T}\rightarrow 0$ as $T\rightarrow\infty$.
\end{lemma}

\begin{proof}
There are two cases to consider.\medskip

\noindent\textit{Case 1: $\chi\equiv 0$.}\medskip

This case corresponds to the boundary components $\partial_{i}M$, $i=1,\ldots,m$. Here the methods of the proof of Theorem \ref{thm2} apply to yield
\begin{align}\label{32}
\begin{split}
& \int_{\partial_{i}\overline{M}_{T}}\frac{1}{8}(q(\overline{N})-\overline{H})(1+v_{T})^{2}+\frac{1}{8}\widehat{H}(1+v_{T})^{4}\\
\geq &\left(\frac{1-\vartheta_{T}}{2}\right)\sqrt{\frac{\pi}{|\partial_{i}\overline{M}_{T}|_{\overline{g}}}}\int_{\partial_{i}\overline{M}_{T}}(1+v_{T})^{2},
\end{split}
\end{align}
for some constants $\vartheta_{T}\rightarrow 0$ as $T\rightarrow\infty$.\medskip

\noindent\textit{Case 2: $\chi$ does not vanish identically.}\medskip

This case corresponds to the boundary components $\partial_{i}M$, $i=m+1,\ldots,n$. For each such component there is a set of positive measure $\Omega_{i}\subset\partial_{i}M$ on which
$\chi_{T_{j}}\geq 2\varepsilon>0$ for a subsequence of heights $T_{j}\rightarrow\infty$. Let $(\overline{M}_{T}-\overline{M}_{T'})\cap\Omega_{i}$ denote the portion
of $(\overline{M}_{T}-\overline{M}_{T'})_{i}$ which, after projection onto the vertical cylinder over $\partial_{i}M$, corresponds with $\Omega_{i}\times(T',T)$. Similarly let
$\partial_{i}\overline{M}_{T}\cap\Omega_{i}$ denote the portion of $\partial_{i}\overline{M}_{T}$ which, after projection onto the vertical cylinder over $\partial_{i}M$, corresponds with
$\Omega_{i}\times \{T\}$. Since $|q|_{\overline{g}}$ is uniformly bounded in $C^{1}$, there is a $\delta>0$ independent of $j$, such that
$|q|_{\overline{g}}\geq \varepsilon$ on $(\overline{M}_{T_{j}+\delta}-\overline{M}_{T_{j}-\delta})\cap\Omega_{i}$ for each $j$. If $\mathcal{N}(T-T_{0})$ denotes the number of $T_{j}$ in the
interval $(T_{0},T)$, then using the asymptotics of $u_{T}$, it follows that for sufficiently large $T$ and $T_{0}$ we have
\begin{align}\label{33}
\begin{split}
\int_{(\overline{M}_{T}-\overline{M}_{T_{0}})_{i}}|q|_{\overline{g}}^{2}(1+v_{T})^{2}&\geq \int_{(\overline{M}_{T}-\overline{M}_{T_{0}})\cap\Omega_{i}}|q|_{\overline{g}}^{2}(1+v_{T})^{2}\\
&\geq\varepsilon^{2}\sum_{j=0}^{\mathcal{N}(T-T_{0})}\int_{(\overline{M}_{T_{j}+\delta}-\overline{M}_{T_{j}-\delta})\cap\Omega_{i}}(1+v_{T})^{2}\\
&\geq\delta\varepsilon^{2}\mathcal{N}(T-T_{0})\int_{\partial_{i}\overline{M}_{T}\cap\Omega_{i}}(1+v_{T})^{2}.
\end{split}
\end{align}
Note that in the last step in the above sequence of inequalities, the factor $\delta$ can be pulled out in light of \eqref{22.1} and the fact that the metric on $(\overline{M}_{T_{j}+\delta}-\overline{M}_{T_{j}-\delta})\cap\Omega_{i}$
approximates the product metric on $\Omega_{i}\times(T_{j}-\delta,T_{j}+\delta)$. Furthermore, using \eqref{22.1} again yields
\begin{equation}\label{34}
\int_{\partial_{i}\overline{M}_{T}\cap\Omega_{i}}(1+v_{T})^{2}\geq C_{0}\int_{\partial_{i}\overline{M}_{T}}(1+v_{T})^{2},
\end{equation}
for some positive constant $C_{0}$ independent of $T$. Hence
\begin{equation}\label{35}
\int_{(\overline{M}_{T}-\overline{M}_{T_{0}})_{i}}\frac{1}{16}|q|_{\overline{g}}^{2}(1+v_{T})^{2}\geq
C_{1}\mathcal{N}(T-T_{0})\sqrt{\frac{\pi}{|\partial_{i}\overline{M}_{T}|_{\overline{g}}}}\int_{\partial_{i}\overline{M}_{T}}(1+v_{T})^{2}.
\end{equation}
From \cite{SchoenYau} (page 257)
\begin{equation}\label{36}
\widehat{g}_{T}\sim\phi_{i}^{4}(y)(d\rho^{2}+4\lambda_{i}\rho^{2}d\theta^{2})
\end{equation}
for $\rho$ near zero, where $\rho=(2\sqrt{\lambda_{i}})^{-1}e^{\mp2\sqrt{\lambda_{i}}t}$ and $d\theta^{2}$ is the induced metric
on $\partial_{i}M$. Therefore
\begin{equation}\label{37}
\widehat{H}\sim\phi_{i}^{-2}(y)\left(\frac{2}{\rho}\right)\sim 4\sqrt{\lambda_{i}}u_{T}^{-2}=4\sqrt{\lambda_{i}}(1+v_{T})^{-2}.
\end{equation}
By applying \eqref{35}, and using the fact that $q(\overline{N})-\overline{H}$ is uniformly bounded and $\mathcal{N}(T-T_{0})\rightarrow\infty$ as $T\rightarrow\infty$, we then have
\begin{align}\label{38}
\begin{split}
& \int_{(\overline{M}_{T}-\overline{M}_{T_{0}})_{i}}\frac{1}{16}|q|_{\overline{g}}^{2}(1+v_{T})^{2}
+\int_{\partial_{i}\overline{M}_{T}}\frac{1}{8}(q(\overline{N})-\overline{H})(1+v_{T})^{2}+\frac{1}{8}\widehat{H}(1+v_{T})^{4}\\
\geq & C_{1}\mathcal{N}(T-T_{0})\sqrt{\frac{\pi}{|\partial_{i}\overline{M}_{T}|_{\overline{g}}}}\int_{\partial_{i}\overline{M}_{T}}(1+v_{T})^{2}
-C_{2}\int_{\partial_{i}\overline{M}_{T}}(1+v_{T})^{2}\\
\geq &\left(1-C_{3}\mathcal{N}(T-T_{0})^{-1}\right)\sqrt{\frac{\pi}{|\partial_{i}\overline{M}_{T}|_{\overline{g}}}}\int_{\partial_{i}\overline{M}_{T}}(1+v_{T})^{2}
\end{split}
\end{align}
for $T$ sufficiently large so that $C_{1}\mathcal{N}(T-T_{0})\geq 1$.

We may now combine \eqref{30}, \eqref{32}, and \eqref{38} to obtain the desired result.
\end{proof}

\section{Proof of the Main Theorem}\label{sec3}

Consider the manifold $(\widehat{M}_{T},\widehat{g}_{T})$. Along the infinite cylindrical ends over $\partial_{i}M$, $i=m+1,\ldots,n$, the
conformal factor $u_{T}$ decays exponentially fast. Therefore as in \cite{SchoenYau} these ends may be closed by adding a point at infinity. The remaining cylindrical ends, indexed by
$i=1,\ldots,m$, correspond to the boundary components of $\widehat{M}_{T}$ which satisfy the hypotheses of Herzlich's version of the positive mass theorem \cite{Herzlich}. Alternatively,
these boundary components have zero Hawking mass
\begin{equation}\label{39}
M_{H}(\partial_{i}\widehat{M}_{T}):=\sqrt{\frac{|\partial_{i}\widehat{M}_{T}|_{\widehat{g}}}{16\pi}}\left(1-\frac{1}{16\pi}\int_{\partial_{i}\widehat{M}_{T}}\widehat{H}^{2}\right)= 0.
\end{equation}
It follows that the ADM energy $E_{ADM}(\widehat{g}_{T})$ is nonnegative. Combining this with the lower bound for $\mathcal{P}(v_{T})$ yields a lower bound for the ADM energy of $g$, from the formula
\begin{equation}\label{40}
E_{ADM}(g)=E_{ADM}(\overline{g})=E_{ADM}(\widehat{g}_{T})+\pi^{-1}\mathcal{P}(v_{T}).
\end{equation}

We now estimate the positive contributions from both terms on the right-hand side of \eqref{40}. Let us begin with $E_{ADM}(\widehat{g}_{T})$. Consider a connected surface $S$ which encloses
$\partial M$, is area outerminimizing, and has nonnegative Hawking mass all with respect to $\widehat{g}_{T}$. Let $\{S_{\varrho}\}_{\varrho=0}^{\infty}$ be a weak inverse mean curvature flow
(see \cite{HuiskenIlmanen}) emanating from $S=S_{0}$. Then according to Geroch monotonicity \cite{HuiskenIlmanen}
\begin{equation}\label{41}
E_{ADM}(\widehat{g}_{T})\geq\int_{0}^{\infty}\left(\frac{|S_{\varrho}|_{\widehat{g}_{T}}^{1/2}}{(16\pi)^{3/2}}\int_{S_{\varrho}}\widehat{R}_{T}d\theta_{\widehat{g}_{T}}\right)d\varrho,
\end{equation}
where $\widehat{R}_{T}$ is the scalar curvature of $\widehat{g}_{T}$. Let $\widehat{N}_{T}$ denote the unit normal to $S_{\varrho}$ with respect to $\widehat{g}_{T}$, and set $\widehat{E}_{T}^{i}=u_{T}^{-4}\overline{E}^{i}$.
By \eqref{16} $\widehat{R}_{T}=2|\widehat{E}_{T}|_{\widehat{g}_{T}}^{2}$, and with the help of Cauchy-Schwarz, \eqref{10}, and H\"{o}lder's inequality
\begin{align}\label{42}
\begin{split}
\int_{S_{\varrho}}|\widehat{E}_{T}|_{\widehat{g}_{T}}^{2}d\theta_{\widehat{g}_{T}}&\geq\int_{S_{\varrho}}\widehat{g}_{T}(\widehat{E}_{T},\widehat{N}_{T})^{2}d\theta_{\widehat{g}_{T}}\\
&=\int_{S_{\varrho}}\overline{g}(\overline{E},\overline{N})^{2}d\theta_{\overline{g}}\\
&\geq|S_{\varrho}|_{\overline{g}}^{-1}\left(\int_{S_{\varrho}}\overline{g}(\overline{E},\overline{N})d\theta_{\overline{g}}\right)^{2}\\
&=\frac{(4\pi\overline{Q})^{2}}{|S_{\varrho}|_{\overline{g}}}\\
&=\frac{|S_{\varrho}|_{\widehat{g}_{T}}}{|S_{\varrho}|_{\overline{g}}}\frac{(4\pi Q)^{2}}{|S_{\varrho}|_{\widehat{g}_{T}}}.
\end{split}
\end{align}
A basic property of inverse mean curvature flow is that the area of the flow surfaces increases exponentially, in particular $|S_{\varrho}|_{\widehat{g}_{T}}=|S_{0}|_{\widehat{g}_{T}}e^{\varrho}$.
Moreover, if $M_{S}$ denotes the region between spatial infinity and the surface $S$, then
\begin{equation}\label{43}
\frac{|S_{\varrho}|_{\widehat{g}_{T}}}{|S_{\varrho}|_{\overline{g}}}\geq\min_{M_{S}}u_{T}^{4}.
\end{equation}
It follows that
\begin{equation}\label{44}
E_{ADM}(\widehat{g}_{T})\geq\left(\sqrt{\frac{|\partial M|_{g}}{|S|_{\widehat{g}_{T}}}}\min_{M_{S}}u_{T}^{4}\right)\sqrt{\frac{\pi}{|\partial M|_{g}}}Q^{2}.
\end{equation}
Since this inequality is true for all surfaces $S$ which enclose
$\partial M$, are area outerminimizing, and have nonnegative Hawking mass all with respect to $\widehat{g}_{T}$, we then have
\begin{equation}\label{45}
E_{ADM}(\widehat{g}_{T})\geq\sigma_{2,T}\sqrt{\frac{\pi}{|\partial M|_{g}}}Q^{2}
\end{equation}
where
\begin{equation}\label{46}
\sigma_{2,T}=\sup_{S}\sqrt{\frac{|\partial M|_{g}}{|S|_{\widehat{g}_{T}}}}\min_{M_{S}}u_{T}^{4}.
\end{equation}
Note that the set of surfaces $S$ which have the above desired properties is nonempty. To see this we may simply start an inverse mean curvature flow from one of the
boundary components $\partial_{i}\widehat{M}_{T}$, $i=1,\ldots,m$, then for sufficiently large $\varrho$, each of the flow surfaces $S_{\varrho}$ encloses $\partial M$,
is area outerminimizing, and has nonnegative Hawking mass. In particular, $\sigma_{2,T}$ is strictly positive.

The positive contribution from $\mathcal{P}(v_{T})$ will now be estimated. Suppose that
$\mathcal{P}(v_{T})\leq\eta\sum_{i=1}^{n}\sqrt{\pi|\partial_{i}\overline{M}_{T}|_{\overline{g}}}$
for some positive constant $\eta$.  Then by Lemma \ref{lemma1}
\begin{equation}\label{47}
\int_{\overline{M}_{T}}\frac{1}{4}|\overline{\nabla}v_{T}|^{2}
+\left(\frac{1-\vartheta_{T}}{2}\right)
\sum_{i=1}^{n}\sqrt{\frac{\pi}{|\partial_{i}\overline{M}_{T}|_{\overline{g}}}}
\int_{\partial_{i}\overline{M}_{T}}(1+v_{T})^{2}
\leq\eta\sum_{i=1}^{n}\sqrt{\pi|\partial_{i}\overline{M}_{T}|_{\overline{g}}}.
\end{equation}
However by Young's inequality
\begin{equation}\label{48}
(1+v_{T})^{2}\geq 1-\frac{1}{\delta}+(1-\delta)v_{T}^{2}
\end{equation}
for any $\delta>0$, and therefore
\begin{align}\label{49}
\begin{split}
&\int_{\overline{M}_{T}}\frac{1}{4}|\overline{\nabla}v_{T}|^{2}+(1-\delta)
\left(\frac{1-\vartheta_{T}}{2}\right)\sum_{i=1}^{n}\sqrt{\frac{\pi}
{|\partial_{i}\overline{M}_{T}|_{\overline{g}}}}\int_{\partial_{i}\overline{M}_{T}}v_{T}^{2}\\
\leq &(\eta-\frac{1}{2}(1-\delta^{-1})(1-\vartheta_{T}))
\sum_{i=1}^{n}\sqrt{\pi|\partial_{i}\overline{M}_{T}|_{\overline{g}}}.
\end{split}
\end{align}
The left-hand side is nonnegative if
$\delta-1\leq\sigma_{1,T}$ where
\begin{equation}\label{50}
\sigma_{1,T}=
\frac{\int_{\overline{M}_{T}}|\overline{\nabla}v_{T}|^{2}}{2(1-\vartheta_{T})
\sum_{i=1}^{n}\sqrt{\frac{\pi}{|\partial_{i}\overline{M}_{T}|_{\overline{g}}}}
\int_{\partial_{i}\overline{M}_{T}}v_{T}^{2}}.
\end{equation}
It follows that $\eta\geq\delta^{-1}(\delta-1)(1-\vartheta_{T})/2$ for all such
$\delta$.  In particular by choosing $\delta=1+\sigma_{1,T}$ we
conclude that
\begin{equation}\label{51}
\mathcal{P}(v_{T})\geq\frac{\sigma_{1,T}(1-\vartheta_{T})}{2(1+\sigma_{1,T})}\sum_{i=1}^{n}\sqrt{\pi
|\partial_{i}\overline{M}_{T}|_{\overline{g}}}.
\end{equation}

The combination of \eqref{40}, \eqref{45}, and \eqref{51} now produces
\begin{equation}\label{52}
E_{ADM}(g)\geq\frac{\sigma_{1,T}(1-\vartheta_{T})}{2(1+\sigma_{1,T})}\sum_{i=1}^{n}\sqrt{\frac{|\partial_{i}
\overline{M}_{T}|_{\overline{g}}}{\pi}} +\sigma_{2,T}\sqrt{\frac{\pi}{|\partial M|_{g}}}Q^{2}.
\end{equation}

\begin{theorem}\label{thm3}
After possibly passing to a subsequence, $u_{T}\rightarrow u$ in $C^{\infty}_{loc}(\overline{M})$ as $T\rightarrow\infty$, where $u$
is the unique solution of boundary value problem \eqref{19}, \eqref{19.1}.
\end{theorem}

\begin{proof}
Together \eqref{31}, \eqref{40}, and \eqref{45} show that the sequence of
functions $\{u_{T}\}$ is uniformly bounded in
$W^{1,2}_{loc}(\overline{M})$. Thus with the help of elliptic
estimates and Sobolev embeddings, a subsequence converges on
compact subsets to a smooth uniformly bounded solution
$u_{\infty}$ of
\begin{equation}\label{53}
\overline{\Delta}u_{\infty}-\frac{1}{8}\overline{R}u_{\infty}+\frac{1}{4}|\overline{E}|^{2}_{\overline{g}}u_{\infty}=0\text{
}\text{ }\text{ on }\text{ }\text{ }\overline{M},\text{ }\text{
}\text{ }\text{
}u_{\infty}=1+\frac{A_{\infty}}{|x|}+O(|x|^{-2})\text{ }\text{
}\text{ as }\text{ }\text{ }|x|\rightarrow\infty.
\end{equation}
Moreover since $\overline{M}$ approximates a cylinder on regions
where it blows-up, comparison with a bounded solution of the same
equation on the cylinder (as is done in \cite{SchoenYau}) shows that
$u_{\infty}(x)\rightarrow 0$ as $x\rightarrow\partial M$; in fact
the decay rate is of exponential strength. Thus $u_{\infty}$ satisfies boundary
value problem \eqref{19}, \eqref{19.1}, and therefore must coincide with the unique solution
to this problem $u=u_{\infty}$.
\end{proof}

Theorem \ref{thm3} shows that after passing to a subsequence, $\sigma_{1,T}\rightarrow\sigma_{1}$ and $\sigma_{2,T}\rightarrow\sigma_{2}$
as $T\rightarrow\infty$. Theorem \ref{thm1} now follows from \eqref{52}.

Lastly we analyze what happens when equality occurs in Theorem \ref{thm1}.  By slightly modifying the arguments presented, we find that equality in \eqref{17} implies that
\begin{equation}\label{54}
\int_{\overline{M}}|\overline{\nabla}u|^{2}=0,
\end{equation}
and therefore $u$ must be constant.  However this is
impossible since
\begin{equation}\label{55}
u(x)\rightarrow\begin{cases}
1 & \text{as $|x|\rightarrow\infty$},\\
0 & \text{as $x\rightarrow\partial M$}.
\end{cases}
\end{equation}
We conclude that the case of equality cannot occur.

\section{Further Properties of the Constant $\sigma_{1}$}\label{sec4}

In this section it will be shown how the constant $\sigma_{1}$ of Theorem \ref{thm1} may be redefined as an infimum over conformal factors
which satisfy appropriate asymptotics. We also describe how the area of each boundary component $\partial_{i}M$, naturally arises and makes a definite contribution
to the right-hand side of \eqref{17}. 

First observe that a slight improvement of the estimate in Lemma \ref{lemma1} is possible, by utilizing all of the terms in \eqref{30}, namely
\begin{align}\label{56}
\begin{split}
\mathcal{P}(v_{T})&\geq \int_{\overline{M}_{T}}\frac{1}{4}|\overline{\nabla}v_{T}|^{2}+\left(\frac{1}{8}(\mu-|J|_{g})+\frac{1}{20}|q|_{\overline{g}}^{2}\right)(1+v_{T})^{2}\\
&+\left(\frac{1-\vartheta_{T}}{2}\right)
\sum_{i=1}^{n}\sqrt{\frac{\pi}{|\partial_{i}\overline{M}_{T}|_{\overline{g}}}}
\int_{\partial_{i}\overline{M}_{T}}(1+v_{T})^{2}.
\end{split}
\end{align}
By following the arguments in Section \ref{sec3}, we obtain a lower bound of the form
\begin{equation}\label{57}
\mathcal{P}(v)\geq\frac{\overline{\sigma}_{1,T}(1-\vartheta_{T})}{2(1+\overline{\sigma}_{1,T})}\sum_{i=1}^{n}\sqrt{\pi
|\partial_{i}\overline{M}_{T}|_{\overline{g}}},
\end{equation}
where
\begin{equation}\label{58}
\overline{\sigma}_{1,T}=
\frac{\int_{\overline{M}_{T}}|\overline{\nabla}v_{T}|^{2}+\left(\frac{1}{2}(\mu-|J|_{g})+\frac{1}{5}|q|_{\overline{g}}^{2}\right)(1+v_{T})^{2}}{2(1-\vartheta_{T})
\sum_{i=1}^{n}\sqrt{\frac{\pi}{|\partial_{i}\overline{M}_{T}|_{\overline{g}}}}
\int_{\partial_{i}\overline{M}_{T}}v_{T}^{2}}.
\end{equation}
It follows that if $\overline{\sigma}_{1,T}\rightarrow\overline{\sigma}_{1}$ then
\begin{align}\label{59}
\begin{split}
\overline{\sigma}_{1}
&=\left(\sum_{i=1}^{n}\sqrt{4\pi|\partial_{i}M|_{g}}\right)^{-1}\int_{\overline{M}}|\overline{\nabla}u|^{2}+\left(\frac{1}{2}(\mu-|J|_{g})+\frac{1}{5}|q|_{\overline{g}}^{2}\right)u^{2}\\
&\geq\left(\sum_{i=1}^{n}\sqrt{4\pi|\partial_{i}M|_{g}}\right)^{-1}\inf_{w}\int_{\overline{M}}|\overline{\nabla}w|^{2}+\left(\frac{1}{2}(\mu-|J|_{g})+\frac{1}{5}|q|_{\overline{g}}^{2}\right)w^{2},
\end{split}
\end{align}
where the infimum is taken over all smooth functions $w$ satisfying the following asymptotics
\begin{equation}\label{60}
w\sim e^{\mp\sqrt{\kappa_{i}}t}\text{ }\text{ }\text{ as }\text{ }\text{ }x\rightarrow\partial_{i}M,\text{ }\text{ }\text{ }\text{ }w\rightarrow 1\text{ }\text{ }\text{ as }\text{ }\text{ }
|x|\rightarrow\infty,
\end{equation}
for some positive constants $\kappa_{i}\leq2\lambda_{i}$; here $\lambda_{i}$, as in Section \ref{sec2}, is the first eigenvalue of
$\Delta-\frac{1}{8}K+\frac{1}{4}(E\cdot n)^{2}$ on $\partial_{i}M$. Therefore we may redefine the constant $\sigma_{1}$ appearing in Theorem \ref{thm1} to be given by the infimum on the right-hand
side of \eqref{59}.

The infimum is realized by the unique (positive) solution of the equation
\begin{equation}\label{61}
\overline{\Delta}w-\left[\frac{1}{2}(\mu-|J|_{g})+\frac{1}{5}|q|_{\overline{g}}^{2}\right]w=0\text{ }\text{ }\text{ on }\text{ }\text{ }\overline{M},
\end{equation}
with asymptotics
\begin{equation}\label{62}
e^{\pm\sqrt{\gamma_{i}}t}w\rightarrow\psi_{i}\text{ }\text{ }\text{ as }\text{ }\text{ }x\rightarrow\partial_{i}M,\text{ }\text{ }\text{ }\text{ }w\rightarrow 1\text{ }\text{ }\text{ as }\text{ }\text{ }
|x|\rightarrow\infty,
\end{equation}
where $\gamma_{i}$ is the first eigenvalue and $\psi_{i}$ the corresponding eigenfunction for the operator $\Delta-\left[\frac{1}{2}(\mu-|J|_{g})+\frac{1}{5}\chi^{2}\right]$ on $\partial_{i}M$, with $\chi$ defined as in Section \ref{sec2}. Note that this is consistent with the fact that $\gamma_{i}\leq2\lambda_{i}$. To see this, observe that for any $\xi\in C^{\infty}_{c}(\overline{M})$
\begin{equation}\label{63}
\int_{\overline{M}}\left(-\overline{R}+2|\overline{E}|^{2}_{\overline{g}}+2(\mu-|J|_{g})+|q|_{\overline{g}}^{2}\right)\xi^{2}
\leq\int_{\overline{M}}(-|q|_{\overline{g}}^{2}+2\overline{div}(q))\xi^{2}
\leq\int_{\overline{M}}4|\overline{\nabla}\xi|^{2},
\end{equation}
from which it follows that
\begin{equation}\label{64}
4\int_{\overline{M}}\left[|\overline{\nabla}\xi|^{2}+\left(\frac{1}{2}(\mu-|J|_{g})+\frac{1}{4}|q|_{\overline{g}}^{2}\right)\xi^{2}\right]
\leq8\int_{\overline{M}}\left[|\overline{\nabla}\xi|^{2}+\left(\frac{1}{8}\overline{R}-\frac{1}{4}|\overline{E}|^{2}_{\overline{g}}\right)\xi^{2}\right].
\end{equation}
By translating the Jang surface in the $t$-direction as in \cite{SchoenYau} (page 254), we find
\begin{equation}\label{65}
\int_{\partial_{i}M}\left[|\nabla\varphi|^{2}+\left(\frac{1}{2}(\mu-|J|_{g})+\frac{1}{4}\chi^{2}\right)\varphi^{2}\right]
\leq2\int_{\partial_{i}M}\left[|\nabla\varphi|^{2}+\left(\frac{1}{8}K-\frac{1}{4}(E\cdot n)^{2}\right)\varphi^{2}\right]
\end{equation}
for any $\varphi\in C^{\infty}(\partial_{i}M)$, so that $\gamma_{i}\leq2\lambda_{i}$.

Next we describe (heuristically) how the area of each boundary component $\partial_{i}M$, naturally arises and makes a definite contribution
to the right-hand side of \eqref{17}. This is primarily a consequence of the cylindrical geometry of the Jang surface near the horizon.
Recall that as in \eqref{51}, the goal is to obtain a lower bound for $\mathcal{P}(v_{T})$ and then let $T\rightarrow\infty$. Observe that
since $(\overline{M}-\overline{M}_{t_{0}})_{i}$ approximates a cylinder for sufficiently large $t_{0}$, it follows that
\begin{equation}\label{66}
4\mathcal{P}(v)\geq\int_{\overline{M}}|\overline{\nabla}u|^{2}\geq\frac{1}{2}\sum_{i=1}^{n}\int_{(\overline{M}-\overline{M}_{t_{0}})_{i}}(\partial_{t}u)^{2}.
\end{equation}
Furthermore, since $u\sim e^{\mp\sqrt{\lambda_{i}}(t-t_{0})}\phi_{i}$ where $\phi_{i}$ is the principal eigenfunction of the operator
$\Delta-\frac{1}{8}K+\frac{1}{4}(E\cdot n)^{2}$ on $\partial_{i}M$, normalized so that $\parallel\phi_{i}\parallel_{L^{2}}^{2}=|\partial_{i}M|_{g}$, we find that
\begin{equation}\label{67}
\mathcal{P}(v)\geq\sum_{i=1}^{n}c_{i}\left(\int_{\partial_{i}M}\phi_{i}^{2}\right)\left(\int_{t_{0}}^{\infty}\lambda_{i}e^{-2\sqrt{\lambda_{i}}(t-t_{0})}\right)
=\sum_{i=1}\frac{c_{i}}{2}\sqrt{\lambda_{i}}|\partial_{i}M|_{g}
\end{equation}
where the constants $c_{i}>0$ depend on $u$. As $\lambda_{i}$ is the principal eigenvalue for a self-adjoint elliptic operator on the 2-sphere, it should behave similarly
to the principal eigenvalue of the Laplacian in that $\lambda_{i}\sim|\partial_{i}M|_{g}^{-1}$. Thus we find a natural contribution to the right-hand side of \eqref{17}, from
each boundary component, in the form of $\sqrt{|\partial_{i}M|_{g}}$.

Another somewhat more vague approach to arrive at the same intuitive conclusion, is to realize that $\mathcal{P}(v)$ is related to the electrostatic capacity of $\partial M$,
which in turn is related to $\sqrt{|\partial M|_{g}}$. There are several well-known results, and also conjectures (of P\'{o}lya and Szeg\"{o}), concerning the relationship of
capacity to the square root of boundary area in Euclidean space. It remains to be seen exactly how these generalize to a Riemannian manifold, although one result in this
direction may be found in \cite{BrayMiao}.

\section{Appendix: The Uncharged Case}\label{sec5}

In this section we make clear how the arguments above correct issues associated with the uncharged Penrose-like inequality
discussed in \cite{Khuri}. Recall that two errors were pointed out in the erratum \cite{Khuri1}. The first concerns the constant
$\sigma$ in the statement of Theorem 1.2 \cite{Khuri}; namely, this constant is in fact zero.  The second error concerns Lemma 2.2 \cite{Khuri},
in that the quantity $\overline{H}-q(\overline{N})$ may not necessarily approach
zero as $r\rightarrow 0$. By setting the electric field $E=0$ in the results of the present paper, both problems are resolved; in particular, Lemma 2.2 \cite{Khuri} is not needed. However, for the convenience of the reader, we explicitly carry out the revised proofs below for the uncharged case.

Let $\chi_{T}(y)$ denote the one parameter family of functions defined on a given boundary component $\partial_{i}M$ as the restriction
of $|q|_{\overline{g}}$ to $\partial_{i}\overline{M}_{T}$. According to the parametric estimates for the Jang equation \cite{SchoenYau}, the sequence of
functions $\chi_{T}$ is uniformly bounded and equicontinuous. Therefore after passing to a subsequence (still denoted by $\chi_{T}$ for convenience)
we have that $\chi_{T}\rightarrow\chi$ as $T\rightarrow\infty$, for some continuous function $\chi$.
There are two cases to consider, namely, case 1 when $\chi$ vanishes identically (in which case the conclusion of Lemma 2.2 holds), and case 2 when $\chi$ does not vanish identically.

Before considering both cases, we construct an appropriate conformal factor. In the work of Schoen and Yau [2] existence of a unique solution to the following boundary value problem was
established:
\begin{equation}\label{a.0}
\overline{\Delta}u-\frac{1}{8}\overline{R}u=0\text{
}\text{ }\text{ on }\text{ }\text{ }\overline{M},
\end{equation}
with $u(x)\rightarrow 0$ as $x\rightarrow\partial M$ and $u(x)\rightarrow 1$ as $|x|\rightarrow\infty$.
A slightly modified version of (4.6) in \cite{SchoenYau} shows
that the first eigenvalue, $\eta_{i}$, of the operator $\Delta-\frac{1}{8}K$ on $\partial_{i}M$, is strictly positive (here $K$ denotes Gaussian curvature). As observed by
Schoen and Yau, $u\sim e^{\mp\sqrt{\eta_{i}}t}\zeta_{i}(y)$, that is the conformal factor $u$ is asymptotic to $e^{\mp\sqrt{\eta_{i}}t}\zeta_{i}(y)$ depending on whether
the Jang surface blows up or down, where $\zeta_{1}$ is the corresponding first eigenfunction.

We will slightly perturb $u$ in order to prescribe appropriate boundary conditions
on certain cylindrical ends. For large $T$ and $T_{0}$ ($T>T_{0}$), let
$(\overline{M}_{T}-\overline{M}_{T_{0}})_{i}$ denote the component of $\overline{M}_{T}-\overline{M}_{T_{0}}$ associated with the boundary component $\partial_{i}M$.
Let $i=1,\ldots,m$ index the boundary components which fall under case 1, and
let $i=m+1,\ldots,n$ index the boundary components which fall under case 2. Set $\widehat{M}_{T}=\overline{M}-\bigcup_{i=1}^{m}(\overline{M}-\overline{M}_{T})_{i}$,
that is, $\widehat{M}_{T}$ is the Jang surface after the cylindrical ends corresponding to case 1 have been removed. Consider the boundary value problem
\begin{equation}\label{a.1}
\overline{\Delta}u_{T}-\frac{1}{8}\overline{R}u_{T}=0\text{
}\text{ }\text{ on }\text{ }\text{ }\widehat{M}_{T},
\end{equation}
\begin{equation*}
\partial_{\overline{N}}u_{T}+\frac{1}{4}\overline{H}u_{T}=\frac{1}{4}\sqrt{\frac{16\pi}{|\partial_{i}\overline{M}_{T}|_{\widehat{g}_{T}}}}u_{T}^{3}
\text{ }\text{ }\text{ on }\text{ }\text{ }\partial_{i}\overline{M}_{T},\text{ }\text{ }i=1,\ldots,m,
\end{equation*}
\begin{equation*}
u_{T}(x)\rightarrow 0\text{ }\text{ }\text{ as }\text{ }\text{ }x\rightarrow\partial_{i} M,\text{ }\text{ }i=m+1,\ldots,n,
\text{ }\text{ }
u_{T}(x)\rightarrow 1\text{ }\text{ }\text{ as }\text{ }\text{ }|x|\rightarrow\infty,
\end{equation*}
where the unit normal $\overline{N}$ points towards spatial infinity and $\widehat{g}_{T}=u_{T}^{4}\overline{g}$.
The unique solution to this problem exists by Theorem \ref{thm2} (with $E=0$), and it approximates the solution $u$ of \eqref{a.0} for large $T$, as is shown in Theorem \ref{uncharged_case} below. Note that as in \cite{SchoenYau} a separation of variables argument can be used to show that this solution possesses the same asymptotics as $u$, namely
\begin{equation}\label{a.2}
u_{T}\sim e^{\mp\sqrt{\eta_{i}}t}\zeta_{i}(y)
\end{equation}
along the ends corresponding to $\partial_{i}M$, $i=m+1,\ldots,n$.

Multiply equation \eqref{a.1} by $u_{T}=1+v_{T}$ and integrate by parts to obtain
\begin{align*}
\begin{split}
Q(v_{T}):= &\lim_{r\rightarrow\infty}\frac{1}{2}\int_{|x|=r}u_{T}\partial_{\overline{N}}u_{T}\\
\geq &\int_{\overline{M}_{T}}\frac{1}{4}|\overline{\nabla}v_{T}|^{2}
+\left(\pi(\mu-|J|)+\frac{1}{16}|q|_{\overline{g}}^{2}\right)(1+v_{T})^{2}\\
 &+\int_{\partial\overline{M}_{T}}\frac{1}{8}q(\overline{N})(1+v_{T})^{2}+\frac{1}{2}u_{T}\partial_{\overline{N}}u_{T}.
\end{split}
\end{align*}
A standard formula yields
\begin{equation*}
\partial_{\overline{N}}u_{T}=\frac{1}{4}\widehat{H}u_{T}^{3}-\frac{1}{4}\overline{H}u_{T},
\end{equation*}
where $\widehat{H}$ and $\overline{H}$ are the mean curvatures with respect to $\widehat{g}_{T}$ and $\overline{g}$, respectively. It
follows that
\begin{align}\label{a.3}
\begin{split}
Q(v_{T})
\geq&\int_{\overline{M}_{T}}\frac{1}{4}|\overline{\nabla}v_{T}|^{2}
+\left(\pi(\mu-|J|)+\frac{1}{16}|q|_{\overline{g}}^{2}\right)(1+v_{T})^{2}\\
& +\int_{\partial\overline{M}_{T}}\frac{1}{8}(q(\overline{N})-\overline{H})(1+v_{T})^{2}+\frac{1}{8}\widehat{H}(1+v_{T})^{4},
\end{split}
\end{align}
where $\mu$ and $J$ are given by \eqref{2} with $E=0$.\medskip

\noindent\textit{Case 1: $\chi\equiv 0$.}\medskip

This case corresponds to the boundary components $\partial_{i}M$, $i=1,\ldots,m$. Here the conclusion of Lemma 2.2 \cite{Khuri} is valid, and the boundary conditions of
(3.1) in \cite{Khuri} hold. Thus the methods of \cite{Khuri} (as in the proof of Theorem 3.1) apply to yield
\begin{align}\label{a.4}
\begin{split}
& \int_{\partial_{i}\overline{M}_{T}}\frac{1}{8}(q(\overline{N})-\overline{H})(1+v_{T})^{2}+\frac{1}{8}\widehat{H}(1+v_{T})^{4}\\
&\geq \left(\frac{1-\vartheta_{T}}{2}\right)\sqrt{\frac{\pi}{|\partial_{i}\overline{M}_{T}|}}\int_{\partial_{i}\overline{M}_{T}}(1+v_{T})^{2},
\end{split}
\end{align}
where the constants $\vartheta_{T}\rightarrow 0$ as $T\rightarrow\infty$.
\hfill $\Box$ \medskip

\noindent\textit{Case 2: $\chi$ does not vanish identically.}\medskip

Here we may follow the proof in Lemma \ref{lemma1} (of the present paper) directly to obtain
\begin{align}\label{a.6}
\begin{split}
& \int_{(\overline{M}_{T}-\overline{M}_{T_{0}})_{i}}\frac{1}{16}|q|_{\overline{g}}^{2}(1+v_{T})^{2}
+\int_{\partial_{i}\overline{M}_{T}}\frac{1}{8}(q(\overline{N})-\overline{H})(1+v_{T})^{2}+\frac{1}{8}\widehat{H}(1+v_{T})^{4}\\
\geq& \left(1-C\mathcal{N}(T-T_{0})^{-1}\right)\sqrt{\frac{\pi}{|\partial_{i}\overline{M}_{T}|}}\int_{\partial_{i}\overline{M}_{T}}(1+v_{T})^{2},
\end{split}
\end{align}
for $T$ sufficiently large.
\hfill $\Box$ \medskip

By combining \eqref{a.3}, \eqref{a.4}, and \eqref{a.6} we conclude that
\begin{equation}\label{a.7}
Q(v_{T})\geq \int_{\overline{M}_{T}}\frac{1}{4}|\overline{\nabla}v_{T}|^{2}
+\left(\frac{1-\vartheta_{T}}{2}\right)
\sum_{i=1}^{n}\sqrt{\frac{\pi}{|\partial_{i}\overline{M}_{T}|}}
\int_{\partial_{i}\overline{M}_{T}}(1+v_{T})^{2}.
\end{equation}

\begin{theorem}\label{uncharged_case}
Let $(M,g,k)$ be an asymptotically
flat initial data set for the Einstein equations satisfying the
dominant energy condition.  If the boundary consists
of an outermost apparent horizon with components
$\partial_{i}M$ having area $|\partial_{i}M|$,
$i=1,\ldots,n$, then
\begin{equation*}
E_{ADM}(g)\geq\frac{\sigma}{2(1+\sigma)}\sum_{i=1}^{n}\sqrt{\frac{|\partial_{i}
M|}{\pi}}
\end{equation*}
where
\begin{equation*}
\sigma=\left(\sum_{i=1}^{n}\sqrt{4\pi|\partial_{i}M|}\right)^{-1}\parallel\overline{\nabla}
u\parallel^{2}_{L^{2}(\overline{M})}.
\end{equation*}
\end{theorem}

\begin{proof}
Consider the manifold $(\widehat{M}_{T},\widehat{g}_{T})$. Along the infinite cylindrical ends over $\partial_{i}M$, $i=m+1,\ldots,n$, the
conformal factor $u_{T}$ decays exponentially fast. Therefore as in \cite{SchoenYau} these ends may be closed by adding a point at infinity. The remaining cylindrical ends, indexed by
$i=1,\ldots,m$, correspond to the boundary components of $\widehat{M}_{T}$ which satisfy the hypotheses of Herzlich's version of the positive mass theorem (Theorem 1.4 of \cite{Khuri}).
It follows that $E_{ADM}(\widehat{g}_{T})$ is nonnegative. At this point, we can apply \eqref{a.7} and follow the same procedure as in section $\S 4$ of \cite{Khuri} to obtain
\begin{equation*}
E_{ADM}(g)\geq\frac{\sigma_{T}(1-\vartheta_{T})}{2(1+\sigma_{T})}\sum_{i=1}^{n}\sqrt{\frac{|\partial_{i}\overline{M}_{T}|}{\pi}},
\end{equation*}
where
\begin{equation*}
\sigma_{T}=\frac{\int_{\overline{M}_{T}}|\overline{\nabla}v_{T}|^{2}}
{2(1-\vartheta_{T})\sum_{i=1}^{n}\sqrt{\frac{\pi}{|\partial_{i}\overline{M}_{T}|}}\int_{\partial_{i}\overline{M}_{T}}v_{T}^{2}}.
\end{equation*}
Moreover, the same arguments in the paragraph after (4.4) in \cite{Khuri} show that $u_{T}\rightarrow u_{\infty}$, where $u_{\infty}$ is a bounded solution of boundary value problem \eqref{a.0}, and
$\sigma_{T}\rightarrow\sigma_{\infty}$. Since there is a unique bounded solution of \eqref{a.0}, it follows that $u_{\infty}=u$ and $\sigma_{\infty}=\sigma$, from which we obtain the desired result.
\end{proof}

Note that in \cite{Khuri}, in the last step, we replaced $\sigma_{\infty}$ with a different definition of $\sigma$ that employed an infimum. The purpose of this replacement was solely to give our
result an expression akin to that in Herzlich's Penrose-like inequality \cite{Herzlich}. Alas, this misguided sense of aesthetics resulted in the error mentioned at the beginning of the appendix.

Lastly, we mention that the definition of $\sigma$ in the current theorem is strictly positive, and $\sigma$ is dimensionless making it independent of the area of $\partial M$.

\end{document}